\documentclass[11pt]{article}

\usepackage[margin=1in]{geometry}
\usepackage{amsmath,amssymb,amsthm,mathtools}
\usepackage{enumitem}
\usepackage{microtype}
\usepackage{sn-preamble}

\newcommand{\dc}{\operatorname{dc}}

\newcommand{\sgn}{\operatorname{sign}}
\newcommand{\stat}{\textsf{STAT}}
\newcommand{\opt}{\textsf{OPT}}

\title{A Separation Between Distribution-Free SQ Learning and Dimension Complexity}
\author{Shyamal Patel\thanks{Email: \texttt{shyamalpatelb@gmail.com}} \\ UT Austin}
\date{}

\begin{document}

\pagenumbering{gobble}
\maketitle

\begin{abstract}
We show that there exists a class of boolean functions $\mathcal{C}$ such that $(i)$ there is a distribution-independent statistical query algorithm for learning $\mathcal{C}$ that makes a polynomial number of queries of inverse polynomial tolerance and $(ii)$ for any set of functions $\Phi_1, \dots, \Phi_r$ such that for all $f \in \mathcal{C}$ we can write $f(x) = \sgn \left( \sum_{i = 1}^r w_i \Phi_i(x) \right)$ for some set of weights $w_i \in \mathbb{R}$, we must have that $r \geq n^{\omega(1)}$. This gives a superpolynomial separation between dimension complexity and the query complexity of distribution-free learning in the statistical query model, negatively answering a question of Feldman, Kamath, and Srebro \cite{FKS}.

Our construction $\mathcal{C}$ is a subclass of DNFs, and the proof is a simple consequence of recent progress on agnostically learning conjunctions \cite{DKR, CPS} and the work of Razborov and Sherstov on the sign rank of DNFs \cite{RS}.
\end{abstract}

\section{Introduction}
A popular approach in machine learning and computational learning theory is the kernel method: To learn a class $\mathcal{C}$ of Boolean functions over $\bits^n$, it suffices to find a set of functions $\Phi_1, \dots, \Phi_r$ such that for all $f \in \mathcal{C}$, we can write 
	\[f(x) = \sign \left( \sum_{i =1}^r w_i \Phi_i(x) \right)\]
for some weights $w_i \in \mathbb{R}$. The smallest value $r$ such that functions $\Phi_1, \dots, \Phi_r$ as above exist is called the \emph{dimension complexity} of $\calC$. Since we have efficient algorithms for learning halfspaces, this then yields an efficient algorithm for learning $\mathcal{C}$ in $\poly(r,n)$ time as long as the functions $\Phi_i$ can be computed efficiently \cite{blumer1989learnability, blum1998polynomial}.

In this note, we consider a question of Feldman, Kamath, and Srebro \cite{FKS}, which asks if the converse is true: if a class has large dimension complexity, is it hard to learn in the distribution-free statistical query model (cf. \Cref{sec:sq})? Our main result is a negative resolution of this question.

\begin{theorem}
\label{thm:main}
There exists a class $\mathcal{C}$ of Boolean functions over $\bits^n$ that can be efficiently learned in the distribution-free statistical query model to error $0.01$, but has dimension complexity $n^{\omega(1)}$. 
\end{theorem}

More precisely, our class will have quasi-polynomial dimension complexity. Notably, if we relax the model so the marginal is known to the learner, stronger separations are known \cite{sherstov2008communication}; that said, algorithms are not given such access in the standard statistical query model.

At a high level, $\mathcal{C}$ will correspond to a family of read-once DNFs. A result of Razborov and Sherstov on the sign rank of a related matrix will then imply that this class has large dimension complexity \cite{RS}. On the other hand, recent progress on agnostically learning conjunctions will yield the desired efficient learning algorithms \cite{DKR, CPS}.

\medskip
\textbf{Statement on AI Usage.} The ideas in this note were generated by the author without the use of AI tools.

\section{Preliminaries}
\paragraph{Conventions.}
We denote random variables using bold faced letters, e.g. $\bx$. Boolean functions will be $\{\pm 1\}$-valued. 
We will denote the dimension complexity of a class by $\dc(\calC)$.

\subsection{The Statistical Query Model}
\label{sec:sq}
For completeness, we give a description of the statistical query model of Kearns \cite{kearns1998efficient}. In this model, there is an unknown distribution $\calD$ over $\bits^n$ and unknown function $f \in \mathcal{C}$. Given an error parameter $\eps$, the goal of the learning algorithm is then to output a hypothesis $h$ satisfying 	
\[\Pr_{\bx \sim \calD} [f(\bx) = h(\bx)] \geq 1 - \eps \]		
with probability $0.99$.
To do so, the learning algorithm then gets access to an oracle $\stat_\tau$. The algorithm can query this oracle by choosing a function $\Phi: \bits^n \times \bits \rightarrow [-1,1]$ and a tolerance parameter $\tau \in (0,1)$. The value of $\stat_\tau(\Phi)$ is then an arbitrary real number satisfying 
		\[ \left | \E_{\bx \sim \calD} [ \Phi(\bx,f(\bx)) ] - \stat_\tau(\Phi) \right| \leq \tau.\]

We say that a statistical query algorithm is efficient if it makes $\poly(n)$ queries of tolerance at least $1/\poly(n)$.

We will crucially need the following boosting result

\begin{theorem}[Aslam and Decatur \cite{aslam}]
\label{thm:boosting}
	Suppose that there exists an SQ algorithm that makes $Q$ queries of tolerance at least $\tau$ and outputs a hypothesis such that $\Pr_{\bx \sim \calD} [f(\bx) = h(\bx)] \geq \frac{1}{2} + \gamma$, then there is a statistical query algorithm that makes $O(Q \cdot \poly(\gamma^{-1}, \log(1/\eps)))$ queries of tolerance $\Omega(\tau \eps^2)$ and outputs a hypothesis with error at most $\eps$.
\end{theorem}

\section{The Construction}
Our construction will be parameterized by a value $m$. We partition $[n]$ into $m$ blocks of size $m^2$, $B_1, \dots, B_m$. We then let $\calC$ be the class of all DNF formulas with one term assigned to each block. That is,
\[
   c(x)=T_1(x|_{B_1})\vee T_2(x|_{B_2})\vee\cdots\vee T_m(x|_{B_m}),
\]
where each $T_i$ is a conjunction of literals using only coordinates from
$B_i$. We allow the empty conjunction, which is identically $1$.

Our main goal will then be to show that

\begin{theorem}
\label{thm:main-formal}
The class $\calC$ satisfies

$(i)$ $\dc(\calC) = 2^{\Omega(n^{1/3})}$

$(ii)$ $\calC$ can be learned to error $0.01$ using $2^{\wt{O}(n^{2/9})}$ queries of tolerance $2^{-\wt{O}(n^{2/9})}$.
\end{theorem}

Note that such a theorem would immediately imply \Cref{thm:main} as we can simply ``scale down'' the class to be on $\wt{O}(\log^{9/2}(n))$ variables.

\subsection{Dimension Complexity Lower Bound}
We start by proving the lower bound on the dimension complexity, which follows from a result of Razborov and Sherstov. Before stating the result, we first recall the definition of sign rank.

\begin{definition}[Sign Rank]
	Given a matrix $A \in \mathbb{R}^{N \times N}$, the sign rank of $A$ is the smallest value $r$ such that there exists a matrix $B$ of rank $r$ satisfying $\sign(B_{ij}) = \sign(A_{ij})$ for all $i,j \in [N]$. 
\end{definition}

\begin{theorem}[Razborov--Sherstov \cite{RS}]
\label{thm:rs}
For $m\ge1$, define
\[
   f_m(x,y)=\bigwedge_{i=1}^{m}\;\bigvee_{j=1}^{m^2}(x_{ij}\wedge y_{ij}),
   \qquad x,y\in \bits^{m^3}.
\]
Then the sign matrix $[f_m(x,y)]_{x,y}$ has sign-rank
$2^{\Omega(m)}$.
\end{theorem}

This will then immediately yield our lower bound on the dimension complexity of the class $\calC$. 

\begin{claim}
	$\dc(\calC) = 2^{\Omega(n^{1/3})}$
\end{claim}

Notably, this claim is essentially shown in \cite{RS}, but we include a proof for completeness.

\begin{proof}
	Let
		\[g_y(x) = -f_m(x,y) = \bigvee_{i=1}^{m}\;\bigwedge_{j=1}^{m^2} (\overline{x_{ij}} \lor \overline{y_{ij}})\]
	where $x_{ij}$ denotes the $j$th coordinate in block $B_i$. We can now observe that for any fixed $y \in \bits^n$ we have $g_y(x) \in \calC$. Now let $\phi: \bits^n \rightarrow \R^{\dc(\calC)}$ be such that
		\[g_y(x) = \sign(\langle w_y, \phi(x) \rangle) \]
	for some $w_y \in \mathbb{R}^{\dc(\calC)}$. We then immediately have that the sign rank of $[g_y(x)]_{x,y} = -[f_m(x,y)]_{x,y}$ is at most $\dc(\calC)$. Applying \Cref{thm:rs} then completes the proof.
\end{proof}

\subsection{Statistical Query Upper Bound}
For the upper bound, we rely on recent results for agnostically learning conjunctions \cite{DKR,CPS}. In particular, it will be convenient to use the result of \cite{DKR} as they cast their result in the SQ model. 

\begin{theorem}[Diakonikolas et al. \cite{DKR}]
\label{thm:agnostic-conj}
For every $\eta\in(0,1)$, there exists a distribution-free agnostic SQ learner for conjunctions over $\bits^n$ that uses $2^{n^{1/3} \polylog(n, 1/\eta)}$ queries of tolerance $2^{-n^{1/3}\polylog(n, 1/\eta)}$ in time $2^{n^{1/3} \polylog(n, 1/\eta)}$ and learns to error $\opt + \eta$.
\end{theorem}

We now turn to prove that $\calC$ can be learned with few queries.

\begin{claim}
$\calC$ can be learned to error $0.01$ using $2^{\wt{O}(n^{2/9})}$ queries of tolerance $2^{-\wt{O}(n^{2/9})}$. 
\end{claim}

\begin{proof}
We prove the claim by giving a weak SQ learning algorithm and then apply boosting to finish the proof. Our weak learner is relatively simple. For each $i \in [m]$, we run the agnostic conjunction learner from \Cref{thm:agnostic-conj} with $\eta = 0.01 \cdot m^{-1}$ on examples $(x|_{B_i},y)$ yielding the hypothesis $h_i$. We then output the function among $h_1, \dots, h_m, -1,1$ with the largest advantage, where $-1,1$ denote constant functions. (We evaluate the advantage of each function to additive accuracy $0.01m^{-1}$ with a statistical query.)

Since we run the agnostic learning algorithm on $\bits^{B_i}$, it follows that the procedure runs in time $2^{(m^2)^{1/3} \cdot \polylog(n,m)} = 2^{\wt{O}(n^{2/9})}$ and satisfies the desired query and tolerance bounds. As such, it suffices to show that some function has advantage $0.05 \cdot m^{-1}$. Towards this, let
\[c(x)=T_1(x|_{B_1})\vee T_2(x|_{B_2})\vee\cdots\vee T_m(x|_{B_m})\]

We now assume that
	\[\Pr_{\bx \sim \mathcal{D}}[c(\bx) = 1] \geq \frac{1}{2} - \frac{1}{10m} \]
and 
	\[\Pr_{\bx \sim \mathcal{D}}[c(\bx) = -1] \geq \frac{1}{2} - \frac{1}{10m} \]
as otherwise the $-1$ or $+1$ functions respectively has the desired advantage. It now follows that there exists an $i$ such that
	\[\Pr_{\bx \sim \calD} \left[T_i(\bx|_{B_i}) = 1 \right] \geq \frac{1}{3m} \]
We then conclude that
	\[\Pr_{\bx \sim \calD} \left[T_i(\bx|_{B_i}) = c(\bx) \right] \geq \Pr_{\bx \sim \calD} \left[ c(\bx) = -1 \right] + \Pr_{\bx \sim \calD} \left[T_i(\bx|_{B_i}) = 1 \right] \geq \frac{1}{2} + \frac{1}{5m} \]
and thus that
	\[\Pr_{\bx \sim \calD}[h_i(\bx) \not = c(\bx)] \leq \Pr_{\bx \sim \calD} \left[T_i(\bx|_{B_i}) \not = c(\bx) \right] + \eta \leq \frac{1}{2} - \frac{1}{5m} + \eta \leq \frac{1}{2} - \frac{1}{6m}.\]
This completes the proof of the weak learner and boosting via \Cref{thm:boosting} then yields the claim.
\end{proof}

\begin{flushleft}
\bibliographystyle{alpha}
\bibliography{allrefs}
\end{flushleft}

\end{document}